\documentclass[12pt]{article}

\usepackage{amsfonts}
\usepackage{amsmath}
\usepackage{amsbsy}
\usepackage{graphicx}
\usepackage{amssymb}
\usepackage{amsthm}
\usepackage{bm}
\usepackage{epsfig}
\usepackage{latexsym}
\usepackage{pdflscape}
\numberwithin{equation}{section}
\input amssym.def
\input amssym.tex
\allowdisplaybreaks

\DeclareMathOperator{\Det}{Det}
\DeclareMathOperator{\Ai}{Ai}

\newcommand{\ii}{\mathtt{i}}
\newcommand{\p}{{\bm Q}}
\newcommand{\q}{{\bm P}}
\newcommand{\e}{{E}}

\def\PP{\mathbb{P}}
\def\sgn{\mathrm{sgn}}

\newcounter{aff}

\theoremstyle{plain}
\newtheorem{theorem}{Theorem}[section]
\newtheorem{lemma}[theorem]{Lemma}

\begin{document}

\begin{titlepage}
\begin{center}
{\Large\bf ABJ Fractional Brane from ABJM Wilson Loop}

\bigskip\bigskip
{\large Sho Matsumoto\footnote[1]{\tt sho-matsumoto@math.nagoya-u.ac.jp}
\quad and \quad
Sanefumi Moriyama\footnote[2]{\tt moriyama@math.nagoya-u.ac.jp}
}\\
\bigskip\bigskip
${}^{*,\dagger}$\,{\normalsize\it 
Graduate School of Mathematics, Nagoya University\\
Nagoya 464-8602, Japan} \bigskip\\
${}^{\dagger}$\,{\normalsize\it
Kobayashi Maskawa Institute, Nagoya University\\
Nagoya 464-8602, Japan} \bigskip\\
${}^{\dagger}$\,{\normalsize\it
Yukawa Institute for Theoretical Physics, Kyoto University\\
Kyoto 606-8502, Japan} \bigskip\\

\end{center}

\begin{abstract}
We present a new Fermi gas formalism for the ABJ matrix model.
This formulation identifies the effect of the fractional M2-brane in
the ABJ matrix model as that of a composite Wilson loop operator in
the corresponding ABJM matrix model.
Using this formalism, we study the phase part of the ABJ partition
function numerically and find a simple expression for it.
We further compute a few exact values of the partition function at
some coupling constants.
Fitting these exact values against the expected form of the grand
potential, we can determine the grand potential with exact
coefficients.
The results at various coupling constants enable us to conjecture an
explicit form of the grand potential for general coupling constants.
The part of the conjectured grand potential from the perturbative sum,
worldsheet instantons and bound states is regarded as a natural
generalization of that in the ABJM matrix model, though the membrane
instanton part contains a new contribution.
\end{abstract}
\end{titlepage}

\section{Introduction}\label{intro}
An explicit Lagrangian description of multiple M2-branes \cite{ABJM} 
has opened up a new window to study M-theory or non-perturbative
string theory.
It was proposed that $N$ multiple M2-branes on
${\mathbb C}^4/{\mathbb Z}_k$ are described by ${\mathcal N}=6$
supersymmetric Chern-Simons-matter theory with gauge group
$U(N)\times U(N)$ and levels $k$ and $-k$.
Due to supersymmetry, partition function and vacuum expectation values
of BPS Wilson loops in this theory on $S^3$ were reduced to a matrix
integration \cite{P,KWY,J,HHL}, which is called the ABJM matrix
model.
Here the coupling constant of the matrix model is related to the
level $k$ inversely.

The ABJM matrix model has taught us much about M-theory or stringy
non-perturbative effects.
Among others, we have learned \cite{DMP1} that it reproduces the
$N^{3/2}$ behavior of the degrees of freedom when $N$ multiple
M2-branes coincide, as predicted from the gravity dual \cite{KT}.
Also, as we see more carefully below, it was found in \cite{HMO2}
that all the divergences in the worldsheet instantons are cancelled
exactly by the membrane instantons.
This reproduces the lesson we learned in the birth of M-theory or
non-perturbative strings:
String theory is not just a theory of strings.
It is only after we include non-perturbative branes that string theory
becomes safe and sound.

After the pioneering paper \cite{DMP1} which reproduced the leading
$N^{3/2}$ behavior, the main interest in the study of the ABJM matrix
model was focused on the perturbative sum \cite{FHM,KEK} and instanton
effects \cite{DMP1,DMP2}.
All of the computations in these papers were done in the 't Hooft
limit, $N\to\infty$ with the 't Hooft coupling $\lambda=N/k$ held
fixed, though for approaching to the M-theory regime with a fixed
background, we have to take a different limit.
Namely, we have to consider the limit $N\to\infty$ with the parameter
$k$ characterizing M-theory background fixed \cite{HKPT,O}.
To overcome this problem, in \cite{MP} the matrix model was rewritten,
using the Cauchy determinant formula, into the partition function of a
Fermi gas system with $N$ non-interacting particles, where the Planck
scale is identified with the level: $\hbar=2\pi k$.
This expression separates the roles of $k$ from $N$, which enables us
to take the M-theory limit.
Note that the M-theory limit probes quite different regimes from the
't Hooft limit.
Especially, using the WKB expansion in the M-theory limit, we can
study the $k$ expansion of the membrane instantons
systematically.

Using the Fermi gas formalism, we can also compute several exact
values of the partition function with finite $N$ at some coupling
constants \cite{HMO1,PY}.
We can extrapolate these exact values to the large $N$ regime and
read off the grand potential \cite{HMO2}.
The grand potential reproduces perfectly the worldsheet instanton
effects predicted by its dual topological string theory on local
$\PP^1\times\PP^1$ when instanton number is smaller than $k/2$, though
serious discrepancies appear beyond it.
Namely, the worldsheet instanton part of the grand potential is
divergent at some values of the coupling constant, while the partition
function of the matrix model is perfectly finite in the whole region
of the coupling constant.
By requiring the cancellation of the divergences and the conformance
to the finite exact values of the partition function at these coupling
constants, we can write down a closed expression for the first few
membrane instantons for general coupling constants \cite{HMO2,CM},
which also matches with the WKB expansion.
Furthermore, using the exact values, we can study the bound states of
the worldsheet instantons and the membrane instantons \cite{HMO3}.
We also find that the instanton effects consist only of the
contributions from the worldsheet instantons, the membranes instantons
and their bound states, and no other contributions appear.
Finally in \cite{HMMO} we relate the membrane instanton to the
quantization of the spectral curve of the matrix model, which is
further related to the refined topological strings on local
$\PP^1\times\PP^1$ in the Nekrasov-Shatashivili limit
\cite{MM1,MM2,ACDKV}.

From the exact solvability viewpoints, we could say that the ABJM matrix
model belongs to a new class of solvable matrix models besides that of
the Gaussian ones and that of the original Chern-Simons ones.
As we have seen, this class of matrix models can be rewritten into a
statistical mechanical model using the Cauchy determinant formula and
contains an interesting structure of pole cancellations between
worldsheet instantons and membrane instantons.
The ABJM matrix model is the only example satisfying these properties
so far.

The most direct generalization of the ABJM theory is the ABJ theory
\cite{ABJ} with the inclusion of fractional branes.
It was proposed that ${\mathcal N}=6$ supersymmetric
Chern-Simons-matter theory with gauge group $U(N_1)\times U(N_2)$ and
the levels $k,-k$ describes ${\rm min}(N_1,N_2)$ M2-branes with
$|N_1-N_2|$ fractional M2-branes on ${\mathbb C}^4/{\mathbb Z}_k$.
The partition function and the vacuum expectation values of the
BPS Wilson loops in the ABJ theory are also reduced to matrix models.
Without loss of generality we can assume $M=N_2-N_1\ge 0$ and $k\ge 0$
for expectation values of hermitian operators.
The unitarity constraint requires $M$ to satisfy $0\le M\le k$.

The integration measure of the ABJM matrix model preserves the super
gauge group $U(N|N)$ while that of the ABJ matrix model preserves
$U(N_1|N_2)$ \cite{DT,MPtop}.
In the language of the topological string theory, the ABJM matrix
model corresponds to the background geometry local $\PP^1\times\PP^1$
with two identical Kahler parameters, while the ABJ matrix model
corresponds to a general non-diagonal case.
Hence, the ABJ matrix model is a direct generalization also from this
group-theoretical or topological string viewpoint.

In this paper we would like to study how the nice structures found in
\cite{MP,HMO1,HMO2,HMO3,HMMO} are generalized to the ABJ matrix
model.
We start our project by presenting a Fermi gas formalism for the ABJ
matrix model.
Our formalism shares the same density matrix as that of the ABJM
matrix model and hence the same spectral problem \cite{KM}.
The effects of fractional branes are encoded in a determinant factor
which takes almost the same form as that of the half-BPS Wilson loops
in the ABJM matrix model \cite{HHMO}.

Another interesting Fermi-gas formalism was proposed previously by the
authors of \cite{AHS}.\footnote{There were some points in \cite{AHS}
which need justification.
This is another motivation for our current proposal.
After we finished establishing this new formalism and proceeded to
studying the grand potential, we were informed by M.~Honda of his
interesting work \cite{H}.}
Compared with their formulation, our formalism has an advantage in the
numerical analysis since the density matrix is the same and all the
techniques used previously can be applied here directly.

In the formalism of \cite{AHS}, they found that the formula with
integration along the real axis is only literally valid for
$0\le M\le k/2$.
For $k/2<M\le k$, additional poles get across the real axis and we
need to deform the integration contour to avoid these poles.
Here we find that the same deformation is necessary in our formalism.
Besides, we have pinned down the origin of this deformation in the
change of variables in the Fourier transformation.

We believe that our Fermi gas formalism has also cast a new viewpoint
to the fractional branes.
In string theory, it was known that graviton sometimes puffs up into a
higher-dimensional object, which is called giant graviton \cite{MGST}.
In the gauge theory picture, this object is often described as a
determinant operator.
Our Fermi gas formalism might suggest an interpretation of the
fractional branes in the ABJ theory as these kinds of composite
objects, though the precise identification needs to be elaborated.
Later we will see that the derivation of our Fermi gas formalism
relies on a modification of the Frobenius symbol (see figure
\ref{frobenius}).
Since the hook representation has a natural interpretation as fermion
excitations, this modification can be regarded as shifting the sea
level of the Dirac sea.
This observation may be useful for giving a better interpretation of
our formula.

Using our new formalism we can embark on studying the instanton
effects.
First of all, we compute first several exact or numerical values of the
partition function.
From these studies, we find that the phase part of the partition
function has a quite simple expression.
The grand potential defined by the partition function after dropping
the phase factors
\begin{align}
J_{k,M}(\mu)=\log\biggl(\sum_{N=0}^\infty e^{\mu N}|Z_k(N,N+M)|\biggr),
\label{grandpot}
\end{align}
can be found by fitting the coefficients of the expected instanton
expressions using these exact values.
We have found that they match well with a natural generalization of
the expression for the perturbative sum, the worldsheet instantons and
the bound states of the worldsheet instantons and the membrane
instantons in the ABJM matrix model.
However, the membrane instanton part contains a new kind of
contribution.

Finally, we conjecture that the large chemical potential expansion of
the grand potential is given by
\begin{align}
J_{k,M}(\mu)&=\frac{C_{k}}{3}\mu_{\rm eff}^3+B_{k,M}\mu_{\rm eff}+A_{k}
+\sum_{m=1}^\infty d_{k,M}^{(m)}e^{-4m\mu_{\rm eff}/k}\nonumber\\
&\quad+\sum_{\ell=1}^\infty(-1)^{M\ell}
\biggl(\widetilde b_{k}^{(\ell)}\mu_{\rm eff}
+\widetilde c_{k}^{(\ell)}-\frac{M}{2C_k}e_k^{(\ell)}\biggr)
e^{-2\ell\mu_{\rm eff}}.
\label{largemu}
\end{align}
Here the perturbative coefficients are
\begin{align}
&C_k=\frac{2}{\pi^2k},\quad
B_{k,M}=\frac{1}{3k}+\frac{k}{24}-\frac{M}{2}+\frac{M^2}{2k},\nonumber\\
&A_k=-\frac{1}{6}\log\frac{k}{4\pi}+2\zeta'(-1)
-\frac{\zeta(3)}{8\pi^2}k^2
+\frac{1}{3}\int\frac{dx}{e^{kx}-1}
\Bigl(\frac{3}{x\sinh^2x}-\frac{3}{x^3}+\frac{1}{x}\Bigr),
\end{align}
while the worldsheet instanton coefficients are
\begin{align}
d_{k,M}^{(m)}
=\sum_{g=0}^\infty\sum_{d|m}
\sum_{d_1+d_2=d}
\frac{(-\beta^{-1})^{d_1m/d}(-\beta)^{d_2m/d}n^g_{d_1,d_2}}{m/d}
\Bigl(2\sin\frac{2\pi m}{kd}\Bigr)^{2g-2},
\label{WSd}
\end{align}
with $n^g_{d_1,d_2}$ being the Gopakumar-Vafa invariants of local
${\mathbb P}^1\times{\mathbb P}^1$ and $\beta=e^{-2\pi \ii M/k}$.
Aside from the sign factor $(-1)^{M\ell}$, the membrane instanton
coefficients are the same as in the ABJM case \cite{HMO3,HMMO}
\begin{align}
\widetilde b_k^{(\ell)}&=-\frac{\ell}{2\pi}\sum_{g=0}^\infty
\sum_{d|\ell}\sum_{d_1+d_2=d}
\frac{e^{\ii\pi k\ell(d_1-d_2)/2d}(-1)^g\hat n^g_{d_1,d_2}}{(\ell/d)^2}
\frac{(2\sin\pi k\ell/4d)^{2g}}{\sin\pi k\ell/2d},\nonumber\\
\widetilde c_k^{(\ell)}
&=-k^2\frac{d}{dk}\frac{\widetilde b_k^{(\ell)}}{2\ell k},
\label{bc}
\end{align}
and the bound states are incorporated by
\begin{align}
\mu_{\rm eff}=\mu+\frac{1}{C_k}\sum_{\ell=1}^\infty(-1)^{M\ell}a_k^{(\ell)}e^{-2\ell\mu}.
\end{align}
Note that $\hat n^g_{d_1,d_2}$ in \eqref{bc} is different from
$n^g_{d_1,d_2}$ in \eqref{WSd}.
In terms of the refined topological string invariant
$n^{g_L,g_R}_{d_1,d_2}$, both of them are given as follows \cite{HMMO}:
\begin{align}
n^g_{d_1,d_2}=n^{g,0}_{d_1,d_2},\quad
\hat n^g_{d_1,d_2}=\sum_{g_L+g_R=g}(-1)^gn^{g_L,g_R}_{d_1,d_2}.
\end{align}
It should be noticed that, compared with the ABJM result, our formula
\eqref{largemu} has a non-trivial term multiplied by $e_k^{(\ell)}$,
which is related to $a_k^{(\ell)}$ by
\begin{align}
\sum_{\ell=1}^\infty(-1)^{M\ell}a_k^{(\ell)}e^{-2\ell\mu}
=-\sum_{\ell=1}^\infty(-1)^{M\ell}e_k^{(\ell)}e^{-2\ell\mu_{\rm eff}}.
\end{align}
The coefficients $a_k^{(\ell)}$ and $e_k^{(\ell)}$ are determined from
the quantum mirror map and their explicit form is given in
\cite{HMMO}.
If we restrict ourselves to the case of integral $k$, $a_k^{(\ell)}$
can be read from the following explicit relation between
$\mu_{\rm eff}$ and $\mu$:
\begin{align}
\mu_{\rm eff}=\begin{cases}
\mu-(-1)^{k/2-M}2e^{-2\mu}
{}_4F_3\Bigl(1,1,\frac{3}{2},\frac{3}{2};2,2,2;(-1)^{k/2-M}16e^{-2\mu}\Bigr),&
\mbox{for even $k$},\\
\mu+e^{-4\mu}
{}_4F_3\Bigl(1,1,\frac{3}{2},\frac{3}{2};2,2,2;-16e^{-4\mu}\Bigr),&
\mbox{for odd $k$}.
\end{cases}
\end{align}

The organization of this paper is as follows.
In the next section, we shall first present our Fermi gas formalism
for the partition function and the vacuum expectation values of the
half-BPS Wilson operator.
After giving a consistency check for the conjecture in section
\ref{previousDMP}, we shall proceed to the study of exact and
numerical values of partition function and large chemical potential
expansion of the grand potential using our Fermi gas formalism in
sections \ref{phasedepend} and \ref{grandpotential}.
Finally we conclude this paper by discussing future problems in
section \ref{discussion}.
We present two lemmas in the appendices to support the proof of our
formalism in section \ref{fermigas}.

\section{ABJ fractional brane as ABJM Wilson loop\label{fermigas}}
Let us embark on studying the ABJ matrix model, whose partition
function is given by
\begin{align}
&Z_k(N_1,N_2)
=\frac{(-1)^{\frac{1}{2}N_1(N_1-1)+\frac{1}{2}N_2(N_2-1)}}{N_1!N_2!}
\int\frac{d^{N_1}\mu}{(2\pi)^{N_1}}\frac{d^{N_2}\nu}{(2\pi)^{N_2}}
\nonumber\\
&\qquad\qquad\qquad\times
\left(\frac{\prod_{i<j}2\sinh\frac{\mu_i-\mu_j}{2}
\prod_{a<b}2\sinh\frac{\nu_a-\nu_b}{2}}
{\prod_{i,a}2\cosh\frac{\mu_i-\nu_a}{2}}\right)^2
e^{\frac{\ii k}{4\pi}(\sum_i\mu_i^2-\sum_a\nu_a^2)}.
\label{abjpfdef}
\end{align}
We shall first summarize the main results and prove them in this
section.

If we define the grand partition function by
\begin{align}
\Xi_{k,M}(z)=\sum_{N=0}^\infty z^NZ_k(N,N+M),
\label{grandpf}
\end{align}
it can be expressed in a form very similar to the vacuum expectation
values of the half-BPS Wilson loops in the ABJM matrix model
\cite{HHMO} (see also \cite{BOS,KMSS,GKM}),
\begin{align}
\frac{\Xi_{k,M}(z)}{\Xi_{k,0}(z)}
=\det\bigl(H_{M-p,-M+q-1}(z)\bigr)
_{\begin{subarray}{c} 1 \le p \le M \\ 1 \le q \le M \end{subarray}},
\label{abjpf}
\end{align}
with $H_{p,q}(z)$ defined by
\begin{align}
H_{p,q}(z)=E_{p}(\nu)\circ
\biggl[1+z\p(\nu,\mu)\circ\q(\mu,\nu)\circ\biggr]^{-1}
E_{q}(\nu).
\label{pfcomponent}
\end{align}
Here various quantities
\begin{equation}
\q(\mu,\nu)=\frac{1}{2\cosh\frac{\mu-\nu}{2}},\quad
\p(\nu,\mu)=\frac{1}{2\cosh\frac{\nu-\mu}{2}},\quad
E_j(\nu)=e^{(j+\frac{1}{2})\nu},
\label{eq:f}
\end{equation}
are regarded respectively as matrices or vectors with the indices
$\mu$, $\nu$ and multiplication $\circ$ between them is performed with the
measure
\begin{align}
\int\frac{d\mu}{2\pi}e^{\frac{\ii k}{4\pi}\mu^2},\quad
\int\frac{d\nu}{2\pi}e^{-\frac{\ii k}{4\pi}\nu^2},
\end{align}
as in \cite{HHMO}.

For the vacuum expectation values of the half-BPS Wilson loops in the
ABJ matrix model, we can combine the results of the ABJ partition
function \eqref{abjpf} and the ABJM half-BPS Wilson loop \cite{HHMO}
in a natural way.
As in the ABJM case, the half-BPS Wilson loop in the ABJ matrix model
is characterized by the representation of the supergroup
$U(N_1|N_2)$ whose character is given by the supersymmetric Schur
polynomial
\begin{align}
s_\lambda((e^{\mu_1},\ldots,e^{\mu_{N_1}})/(e^{\nu_1},\ldots,e^{\nu_{N_2}})).
\end{align}
Here $\lambda$ is a partition and we assume that
$\lambda_{N_1+1} \le N_2$ (otherwise, $s_\lambda (x/y)=0$).
The vacuum expectation values are defined by inserting this character
into the partition function
\begin{align}
&\langle s_\lambda\rangle_k(N_1,N_2)
=\frac{(-1)^{\frac{1}{2}N_1(N_1-1)+\frac{1}{2}N_2(N_2-1)}}{N_1!N_2!}
\int\frac{d^{N_1}\mu}{(2\pi)^{N_1}}\frac{d^{N_2}\nu}{(2\pi)^{N_2}}
s_\lambda((e^{\mu_1},\ldots,e^{\mu_{N_1}})/(e^{\nu_1},\ldots,e^{\nu_{N_2}}))
\nonumber\\&\qquad\qquad\qquad\times
\left(\frac{\prod_{i<j}2\sinh\frac{\mu_i-\mu_j}{2}
\prod_{a<b}2\sinh\frac{\nu_a-\nu_b}{2}}
{\prod_{i,a}2\cosh\frac{\mu_i-\nu_a}{2}}\right)^2
e^{\frac{\ii k}{4\pi}(\sum_i\mu_i^2-\sum_a\nu_a^2)}.
\end{align}
Our analysis shows that the grand partition function defined by
\begin{align}
\langle s_\lambda\rangle^{\rm GC}_{k,M}(z)
=\sum_{N=0}^\infty z^N\langle s_\lambda\rangle_k(N,N+M),
\label{grandwl}
\end{align}
is given by
\begin{align}
\frac{\langle s_\lambda\rangle^{\rm GC}_{k,M}(z)}{\Xi_{k,0}(z)}
=\det\Bigl(\bigl(H_{l_p,-M+q-1}(z)\bigr)
_{\begin{subarray}{c}1\le p\le M+r\\1\le q\le M\end{subarray}}\,
\big|\,\bigl(\widetilde H_{l_p,a_q}(z)\bigr)
_{\begin{subarray}{c}1\le p\le M+r\\1\le q\le r\end{subarray}}\Bigr),
\label{abjwl}
\end{align}
where $H_{p,q}(z)$ is the same as that defined in \eqref{pfcomponent}
while $\widetilde H_{p,q}(z)$ is defined by
\begin{align}
\widetilde H_{p,q}(z)=zE_p(\nu)\circ
\bigl[1+z\p(\nu,\mu)\circ\q(\mu,\nu)\circ\bigr]^{-1}
\p(\nu,\mu)\circ E_{q}(\mu).
\end{align}
In \eqref{abjwl}, the arm length $a_q$ and the leg length $l_p$ are
the non-negative integers appearing in the modified Frobenius
notations $(a_1a_2\cdots a_r|l_1l_2\cdots l_{r+M})$ of the Young
diagram $\lambda$.
In the ABJM case, the (ordinary) Frobenius notation
$(a_1a_2\cdots a_r|l_1l_2\cdots l_r)$ of Young diagram
$[\lambda_1\lambda_2\cdots]=[\lambda'_1\lambda'_2\cdots]^{\rm T}$ in the
partition notation was defined by $a_q=\lambda_q-q$,
$l_p=\lambda'_p-p$ with
$r={\rm max}\{s|\lambda_s-s\ge 0\}
={\rm max}\{s|\lambda'_s-s\ge 0\}$
and explained carefully in figure 1 of \cite{HHMO}.
In the ABJ case, we define the modified Frobenius notation 
$(a_1a_2\cdots a_r|l_1l_2\cdots l_{r+M})$
by
\begin{align}
a_q=\lambda_q-q-M,\quad l_p=\lambda'_p-p+M,
\label{FrobeniusSymbol}
\end{align}
with
\begin{align}
r={\rm max}\{s|\lambda_s-s-M\ge 0\}
={\rm max}\{s|\lambda'_s-s+M\ge 0\}-M.
\end{align}
Diagrammatically, the arm length and the leg length are interpreted as
the horizontal and vertical box numbers counted from the shifted
diagonal line.
This is explained further by an example in figure \ref{frobenius}.

Our first observation is the usage of a combination of the Cauchy
determinant formula and the Vandermonde determinant
formula\footnote{We are informed by M.~Honda that this formula already
appeared in \cite{BF}.}
\begin{align}
\frac{\prod_{i<j}^{N_1} (x_i-x_j) \cdot
\prod_{a<b}^{N_2} (y_{a}-y_{b})}
{\prod_{i=1}^{N_1} \prod_{a=1}^{N_2} (x_i +y_{a})} 
= (-1)^{N_1 (N_2-N_1)}
\det \begin{pmatrix}
\tfrac{1}{x_1+y_1} & \cdots & \tfrac{1}{x_1+y_{N_2}} \\
\vdots & \ddots & \vdots \\
\tfrac{1}{x_{N_1}+y_1} & \cdots & \tfrac{1}{x_{N_1}+y_{N_2}} \\
y_1^{N_2-N_1-1} & \cdots & y_{N_2}^{N_2-N_1-1} \\
\vdots & \ddots & \vdots \\
y_1^0 & \dots & y_{N_2}^0 
\end{pmatrix}.
\label{eq:det}
\end{align}
Here on the right hand side, the upper $N_1\times N_2$ submatrix and
the lower $(N_2-N_1)\times N_2$ submatrix are given respectively by
\begin{align}
\biggl(\frac{1}{x_i+y_a}\biggr)
_{\begin{subarray}{c} 1\le i\le N_1\\1\le a\le N_2\end{subarray}},\quad
\bigl(y_a^{N_2-N_1-p}\bigr)
_{\begin{subarray}{c} 1\le p\le N_2-N_1\\1\le a\le N_2\end{subarray}}.
\end{align}
The determinantal formula \eqref{eq:det} can be proved without
difficulty by considering the $N_2\times N_2$ Cauchy determinant and
sending the extra $N_2-N_1$ pieces of $x_i$ to infinity.

Here comes the main idea of our computation.
Without the extra monomials $y_a^{N_2-N_1-p}$, as emphasized in
\cite{MP,HHMO}, the partition function can be rewritten into traces of
powers of the density matrices.
In the study of the ABJM half-BPS Wilson loop \cite{HHMO}, the
monomials of the Wilson loop insertion play the role of the endpoints
in this multiplication of the density matrices.
This can be interpreted as follows:
The partition function is expressed by ``closed strings'' of the
density matrix while the Wilson loops are expressed by ``open
strings''.
This implies that the ABJ partition function, after rewritten by using
\eqref{eq:det}, can also be expressed by powers of the density
matrices with monomials $y_a^{N_2-N_1-p}$ in the both ends, similarly
to the case of the ABJM Wilson loop.
The only problem is to count the combinatorial factors correctly.

We can also prove this relation by counting the combinatorial factors
explicitly.
However, it is easier to present the proof by using various
determinantal formulas.
In the following subsections we shall provide proofs for the results
\eqref{abjpf} and \eqref{abjwl} in this way.
Readers who are not interested in the details of the proofs can accept
the results and jump to section \ref{previousDMP}.

\begin{figure}[tb]
\begin{center}
\includegraphics[trim=0 100 200 100,scale=0.5]{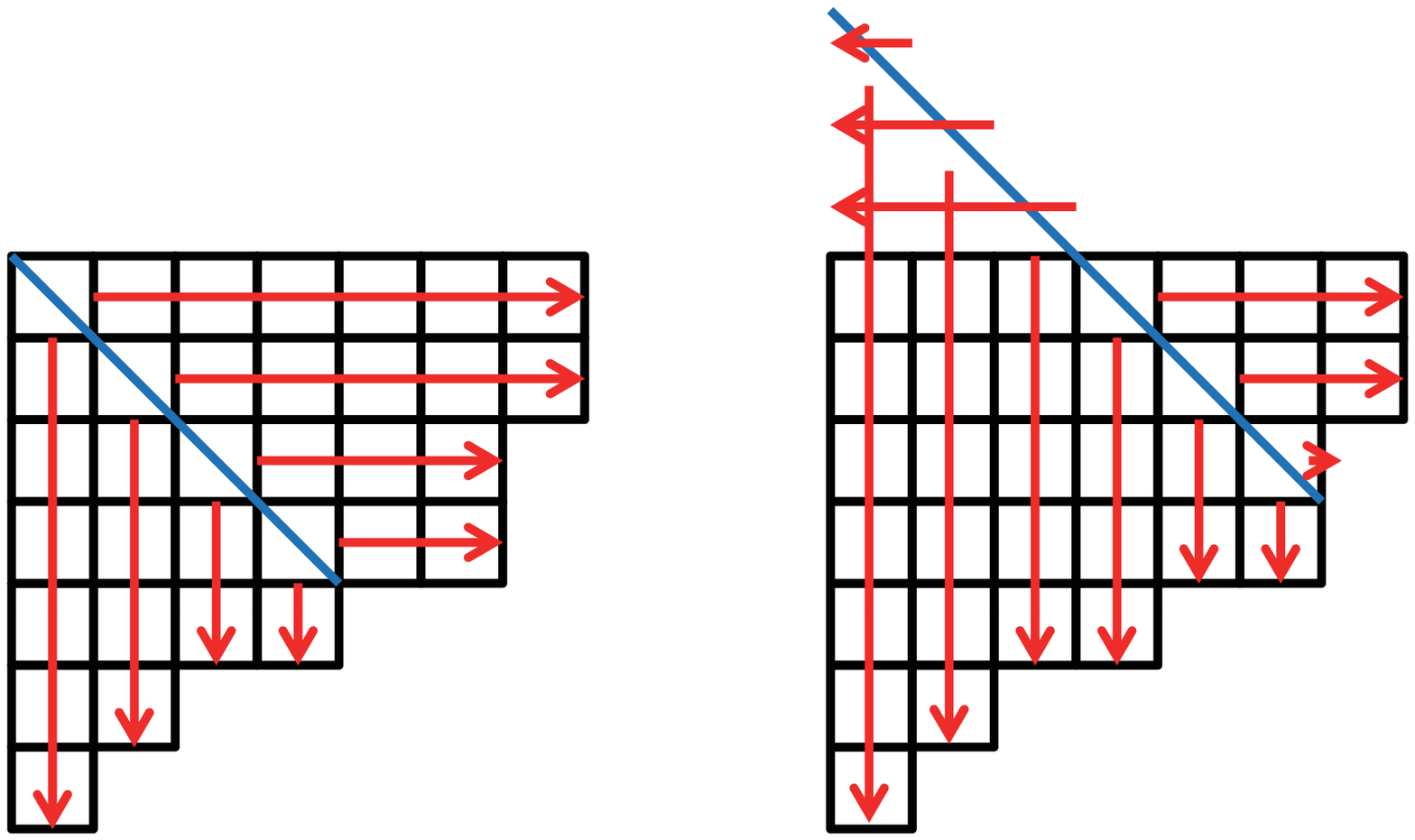}\\
(a)\hspace{5.5cm}(b)
\end{center}
\caption{
Frobenius notation for the ABJM case (a) and for the ABJ case (b).
The same Young diagram
$[\lambda_1\lambda_2\lambda_3\lambda_4\lambda_5
\lambda_6\lambda_7]=[7766421]$
or
$[\lambda'_1\lambda'_2\lambda'_3\lambda'_4\lambda'_5
\lambda'_6\lambda'_7]= [7655442]$
is expressed as $(a_1a_2a_3a_4|l_1l_2l_3l_4)=(6532|6421)$
in the ABJM case while $(a_1a_2a_3|l_1l_2l_3l_4l_5l_6)=(320|975421)$
in the ABJ case $(M=3)$.
It is also convenient to regard the first three horizontal arrows in
(b) as additional arm lengths $(-1,-2,-3)$.
}
\label{frobenius}
\end{figure}

\subsection{Proof of the formula for the partition function}

In this subsection, we shall present a proof for \eqref{abjpf}.
Let us plug $x_i=e^{\mu_i}$ and $y_{a}= e^{\nu_{a}}$ or
$x_i=e^{-\mu_i}$ and $y_a=e^{-\nu_a}$ into \eqref{eq:det}.
Multiplying these two equations side by side, we find
\begin{align}
&(-1)^{\frac{1}{2}N_1(N_1-1)+\frac{1}{2}N_2(N_2-1)}
\left(\frac{\prod_{i<j}2\sinh\frac{\mu_i-\mu_j}{2}
\cdot\prod_{a<b}2\sinh\frac{\nu_a-\nu_b}{2}}
{\prod_{i,a}2\cosh\frac{\mu_i-\nu_a}{2}}\right)^2\nonumber\\
&=\det\begin{pmatrix}
(\q(\mu_i,\nu_j))_{\begin{subarray}{c} 1 \le i \le N_1 \\
1 \le j \le N_2 \end{subarray}} \\
(\e_{M-p}(\nu_j))_{\begin{subarray}{c} 1 \le p \le M \\
1 \le j \le N_2 \end{subarray}}
\end{pmatrix}
\det\begin{pmatrix}
(\p(\nu_j,\mu_i))_{\begin{subarray}{c} 1 \le i \le N_1 \\
1 \le j \le N_2 \end{subarray}} \\
(\e_{-M+p-1}(\nu_j))_{\begin{subarray}{c} 1 \le p \le M \\
1 \le j \le N_2 \end{subarray}}
\end{pmatrix},
\label{eq:two_det}
\end{align}
where $\p$, $\q$ and $\e$ are defined in \eqref{eq:f}.
In order to evaluate the integration of the product \eqref{eq:two_det}
of two $N_2 \times N_2$ determinants, we apply the formula
\eqref{detformula} with $r=0$.
Then we obtain
\begin{align}
& (-1)^{\frac{1}{2}N_1(N_1-1)+\frac{1}{2}N_2(N_2-1)} \frac{1}{N_2 !}
\int \prod_{a=1}^{N_2} \frac{d \nu_a}{2\pi} \
\left( \frac{\prod_{i<j} 2\sinh \frac{\mu_i-\mu_j}{2}
\cdot \prod_{a<b} 2\sinh \frac{\nu_a-\nu_b}{2}}
{\prod_{i,a} 2 \cosh \frac{\mu_i-\nu_a}{2}} \right)^2 
e^{-\frac{\ii k}{4 \pi} \sum_a \nu_a^2} \nonumber \\
&=\det\begin{pmatrix}
((\q \circ \p)(\mu_i,\mu_j))_{1 \le i,j \le N_1} &
((\q \circ \e_{-M+q-1})(\mu_i))_{\begin{subarray}{c}
1 \le i \le N_1 \\ 1 \le q \le M \end{subarray}} \\
((\e_{M-p} \circ \p)(\mu_j))_{\begin{subarray}{c}
1 \le p \le M \\ 1 \le j \le N_1 \end{subarray}} &
(\e_{M-p} \circ \e_{-M+q-1})_{1 \le p,q \le M} 
\end{pmatrix},
\end{align}
where the explicit expression for each component in the determinant is
given by 
\begin{align}
(\q\circ\p)(\mu,\mu')
&=\int\frac{d\nu}{2\pi}\q(\mu,\nu)\p(\nu,\mu')e^{-\frac{\ii k}{4 \pi}\nu^2},&
(\q\circ\e_q)(\mu)
&=\int\frac{d\nu}{2\pi}\q(\mu,\nu)\e_q(\nu)e^{-\frac{\ii k}{4 \pi}\nu^2}, 
\nonumber\\
(\e_p\circ\p)(\mu)
&=\int\frac{d\nu}{2\pi}\e_p(\nu)\p(\nu,\mu)e^{-\frac{\ii k}{4 \pi}\nu^2},&
\e_p\circ\e_q
&=\int\frac{d\nu}{2\pi}\e_p(\nu)\e_q(\nu)e^{-\frac{\ii k}{4 \pi}\nu^2}.
\label{eq:def_ABCD}
\end{align}
Therefore the grand partition function \eqref{grandpf} becomes
\begin{align}
\Xi_{k,M}(z)
&=\sum_{N=0}^\infty 
\frac{z^N}{N!}
\int \prod_{i=1}^N e^{\frac{\ii k}{4 \pi}\mu_i^2}\frac{d \mu_i}{2 \pi}\det
\begin{pmatrix}
((\q\circ\p)(\mu_i,\mu_j))_{N \times N} & 
((\q\circ\e_{-M+q-1})(\mu_i))_{N \times M} \\
((\e_{M-p}\circ\p)(\mu_j))_{M \times N} & 
(\e_{M-p}\circ\e_{-M+q-1})_{M \times M}
\end{pmatrix},
\end{align}
which can be expressed as the Fredholm determinant $\Det$ of the form
\begin{align}
\Xi_{k,M}(z)=\Det\begin{pmatrix}\mathbf{1}+z\q\circ\p&z\q\circ\e\\
\e\circ\p&\e\circ\e\end{pmatrix},
\end{align}
by appendix \ref{fredholm}.
Using the formula
\begin{align}
\Det\begin{pmatrix}A&B\\C&D\end{pmatrix}=\Det A\cdot\Det(D-CA^{-1}B),
\label{inverse}
\end{align}
and simplifying the components by
\begin{align}
\e_p\circ\e_q-z\e_p\circ\p\circ\bigl[1+z\q\circ\p\circ\bigr]^{-1}\q\circ\e_q
=\e_p\circ\bigl[1+z\p\circ\q\circ\bigr]^{-1}\e_q,
\label{eq:symplification}
\end{align}
we finally arrive at \eqref{abjpf}.

\subsection{Proof of the formula for the half-BPS Wilson loop}

In this subsection we shall present a proof for \eqref{abjwl}.
The discussion is parallel to that of the previous subsection.
From the formula due to Moens and Van der Jeugt \cite{MvdJ}, we have
\begin{align}
&s_\lambda((e^{\mu_1},\dots,e^{\mu_{N_1}})/(e^{\nu_1},\dots,e^{\nu_{N_2}}))
\nonumber\\
&=(-1)^{r}\det\begin{pmatrix}
(\q(\mu_i,\nu_j))
_{\begin{subarray}{c} 1 \le i \le N_1 \\ 1 \le j \le N_2 \end{subarray}}
&(\e_{a_q}(\mu_i))
_{\begin{subarray}{c} 1 \le i \le N_1 \\ 1 \le q \le r \end{subarray}} \\
(\e_{l_p}(\nu_j))
_{\begin{subarray}{c} 1 \le p \le M+r \\1 \le j \le N_2 \end{subarray}}
&(0)_{(M+r) \times r}
\end{pmatrix}\bigg/
\det\begin{pmatrix}
(\q(\mu_i,\nu_j))
_{\begin{subarray}{c} 1 \le i \le N_1 \\ 1 \le j \le N_2 \end{subarray}} \\
(\e_{M-p}(\nu_j))
_{\begin{subarray}{c} 1 \le p \le M \\1 \le j \le N_2 \end{subarray}}
\end{pmatrix},
\end{align}
where $(a_1a_2\cdots a_r|l_1l_2\cdots l_{M+r})$
is the modified Frobenius notation of $\lambda$ given in
\eqref{FrobeniusSymbol}.
Combining this determinantal expression with \eqref{eq:two_det},
we have
\begin{align}
&(-1)^{\frac{1}{2}N_1(N_1-1)+\frac{1}{2}N_2(N_2-1)}
s_\lambda ((e^{\mu_1},\dots,e^{\mu_{N_1}})/(e^{\nu_1},\dots,e^{\nu_{N_2}})) \nonumber\\
& \hspace{6cm}\times 
\left(\frac{\prod_{i<j}2\sinh \frac{\mu_i-\mu_j}{2}\cdot
\prod_{a<b}2\sinh\frac{\nu_a-\nu_b}{2}}
{\prod_{i,a}2\cosh \frac{\mu_i-\nu_a}{2}}\right)^2\nonumber\\
&= (-1)^{r}\det \begin{pmatrix}
(\q(\mu_i,\nu_j))
_{\begin{subarray}{c} 1 \le i \le N_1 \\ 1 \le j \le N_2 \end{subarray}}
&(\e_{a_q}(\mu_i))
_{\begin{subarray}{c} 1 \le i \le N_1 \\ 1 \le q \le r \end{subarray}} \\
(\e_{l_p}(\nu_j))
_{\begin{subarray}{c} 1 \le p \le M+r \\1 \le j \le N_2 \end{subarray}}
&(0)_{(M+r) \times r}
\end{pmatrix}
\det\begin{pmatrix}
(\p(\nu_j,\mu_i))_{\begin{subarray}{c} 1 \le i \le N_1 \\
1 \le j \le N_2 \end{subarray}} \\
(\e_{-M+p-1}(\nu_j))_{\begin{subarray}{c} 1 \le p \le M \\
1 \le j \le N_2 \end{subarray}}
\end{pmatrix}.
\end{align}
Integrating this with the formula \eqref{detformula}, we see that
\begin{align}
&(-1)^{\frac{1}{2}N_1(N_1-1)+\frac{1}{2}N_2(N_2-1)}\frac{1}{N_2!}
\int \prod_{a=1}^{N_2} e^{-\frac{\ii k}{4 \pi} \nu_a^2} \frac{d \nu_a}{2\pi}
\left(\frac{\prod_{i<j}2\sinh\frac{\mu_i-\mu_j}{2}
\cdot\prod_{a<b}2\sinh\frac{\nu_a-\nu_b}{2}}
{\prod_{i,a}2\cosh\frac{\mu_i-\nu_a}{2}}\right)^2\nonumber\\
&\hspace{6cm}\times
s_\lambda ((e^{\mu_1},\dots,e^{\mu_{N_1}})/(e^{\nu_1},\dots,e^{\nu_{N_2}}))
\nonumber\\
&=(-1)^r\det\begin{pmatrix}
((\q\circ\p)(\mu_i,\mu_j))
_{\begin{subarray}{c} 1 \le i \le N_1 \\ 1 \le j \le N_1 \end{subarray}}
&((\q\circ\e_{-M+q-1})(\mu_i))
_{\begin{subarray}{c} 1 \le i \le N_1 \\ 1 \le q \le M \end{subarray}} 
&(\e_{a_q}(\mu_i))
_{\begin{subarray}{c} 1 \le i \le N_1 \\ 1 \le q \le r \end{subarray}}
\\
((\e_{l_p}\circ\p)(\mu_j))
_{\begin{subarray}{c} 1 \le p \le M+r \\ 1 \le j \le N_1 \end{subarray}}
& 
(\e_{l_p}\circ\e_{-M+q-1})
_{\begin{subarray}{c} 1 \le p \le M+r \\ 1 \le q \le M \end{subarray}}
& (0)_{(M+r) \times r}
\end{pmatrix}.
\end{align}
Now the definition \eqref{grandwl} of 
$\langle s_\lambda\rangle_{k,M}^{\rm GC}(z)$ and appendix
\ref{fredholm} give
\begin{align}
\langle s_\lambda\rangle_{k,M}^{\rm GC}(z)
&=(-1)^{r}\sum_{N=0}^\infty\frac{z^N}{N!}
\prod_{i=1}^N e^{\frac{\ii k}{4\pi} \mu_i^2} \frac{d \mu_i}{2\pi} 
\nonumber\\ 
&\quad\times\det\begin{pmatrix}
((\q\circ\p)(\mu_i,\mu_j))_{N \times N} & 
(\q\circ\e_{-M+q-1}(\mu_i))_{N \times M} &
(\e_{a_q}(\mu_i))_{N \times r} \\
((\e_{l_p}\circ\p)(\mu_j))_{(M+r) \times N} & 
(\e_{l_p}\circ\e_{-M+q-1})_{(M+r) \times M} & 
(0)_{(M+r) \times r}
\end{pmatrix} \nonumber\\
&= (-1)^{r}
\Det \begin{pmatrix}
1+z \q\circ\p & z\q\circ\e & z\e_a \\
\e_l\circ\p & \e_l\circ\e & 0 
\end{pmatrix}.
\end{align}
Finally, using \eqref{inverse} and \eqref{eq:symplification}, we find 
\begin{align}
&\frac{\langle s_\lambda\rangle_{k,M}^{\rm GC}(z)}
{\Xi_{k,0}(z)} \nonumber\\&
=(-1)^{r} \Det \left[
\e_l\circ\e-\e_l\circ\p\circ(1+z\q\circ\p\circ)^{-1}z\q\circ\e 
\ \bigm| \ 
-\e_l\circ\p\circ(1+z\q\circ\p\circ)^{-1}z\e_a\right]  \nonumber\\
&= \Det \left[
\e_l\circ(1+z\p\circ\q\circ)^{-1}\e\ \bigm| \ 
z\e_l\circ(1+z\p\circ\q\circ)^{-1}\p\circ\e_a \right], 
\end{align}
which is the desired formula \eqref{abjwl}.
In the last determinant, the rows are determined by modified legs
$l_1,l_2,\dots,l_{M+r}$,
whereas the columns are determined by
$(-M,\dots,-2,-1)$ and modified arms 
$a_1,a_2,\dots,a_r$.

\section{Consistency with the previous works}\label{previousDMP}
In the subsequent sections, we shall use our Fermi gas formalism
\eqref{abjpf} to evaluate several values of the partition function and
proceed to confirm our conjecture of the grand potential in
\eqref{largemu}.
However, obviously only the values of the partition function at several
coupling constants are not enough to fix the whole large $\mu$
expansion in \eqref{largemu}.
Hence, before starting our numerical studies, we shall first pause to
study the consistency between our conjecture of the perturbative part
and the worldsheet instanton part in \eqref{largemu} with the
corresponding parts in the 't Hooft expansion \cite{DMP1}.
After fixing the worldsheet instanton contribution, we easily see that
it diverges at some coupling constants.
As in the case of the ABJM matrix model \cite{HMO2}, since the matrix
model is finite for any $(k,M)$ satisfying $0\le M\le k$ (at least
$0\le M\le k/2$, as we shall see in the next section), the divergences
in the worldsheet instantons have to be cancelled by the membrane
instantons and their bound states.
We shall see that, for this cancellation mechanism to work for
$d_{k,M}^{(m)}$, we need to introduce the phase $(-1)^{M\ell}$ for
$\widetilde b_k^{(\ell)}$ and $\widetilde c_k^{(\ell)}$ in
\eqref{largemu}.\footnote{The contents of this section are based on a
note of Sa.Mo.\ during the collaboration of \cite{HMMO}.
Sa.Mo.\ is grateful to the collaborators for various discussions.}

\subsection{Perturbative sum}
The perturbative part of the grand potential in \eqref{largemu}
implies that the perturbative sum of the partition function reads
\begin{align}
Z^{\rm pert}_k(N,N+M)=e^{A_{k}}C_{k}^{-1/3}\Ai[C_{k}^{-1/3}(N-B_{k,M})].
\label{pert}
\end{align}
The argument of the Airy function is proportional to
\begin{align}
\frac{N-B_{k,M}}{k}=\hat\lambda-\frac{1}{3k^2}.
\end{align}
It was noted in \cite{AHHO,DMP1} that the renormalized 't Hooft
coupling constant
\begin{align}
\hat\lambda=\frac{N}{k}-\frac{1}{24},
\end{align}
in the ABJM case has to be modified to
\begin{align}
\hat\lambda=\frac{N_1+N_2}{2k}-\frac{(N_1-N_2)^2}{2k^2}-\frac{1}{24},
\end{align}
in the ABJ case.
We have changed $B_{k,0}$ into $B_{k,M}$ to take care of this
modification.

\subsection{Worldsheet instanton}
Let us see the validity of our conjecture on the worldsheet instanton
$d^{(m)}_{k,M}$.
First note that the worldsheet instanton can be summarized into a
multi-covering formula
\begin{align}
J^{\rm WS}(\mu)=\sum_{g=0}^\infty\sum_{n,d_1,d_2}n^g_{d_1,d_2}
\biggl(2\sin\frac{2\pi n}{k}\biggr)^{2g-2}
\frac{(-e^{-\frac{4\mu}{k}}\beta^{-1})^{nd_1}
(-e^{-\frac{4\mu}{k}}\beta)^{nd_2}}{n}.
\label{kahler}
\end{align}
This naturally corresponds to shifting the two Kahler parameters by
$\pm 2\pi\ii M/k$.

Next, we shall see that the expression of the worldsheet instanton
\eqref{WSd} reproduces the genus-0 free energy of the matrix model
\cite{DMP1}.
As in \cite{HMO2}, the first few worldsheet instanton terms of the
free energy $F_{k,M}=\log Z_{k,M}$ with abbreviation
$Z_{k,M}=Z_k(N,N+M)$ are given by
\begin{align}
F^{{\rm WS}(1)}_{k,M}&=Z^{{\rm WS}(1)}_{k,M},\nonumber\\
F^{{\rm WS}(2)}_{k,M}&=Z^{{\rm WS}(2)}_{k,M}-\frac{1}{2}(Z^{{\rm WS}(1)}_{k,M})^2,
\end{align}
where the partition functions are
\begin{align}
Z^{{\rm WS}(1)}_{k,M}&=d^{(1)}_{k,M}
\frac{\Ai[C_{k}^{-1/3}(N+\frac{4}{k}-B_{k,M})]}
{\Ai[C_{k}^{-1/3}(N-B_{k,M})]},\nonumber\\
Z^{{\rm WS}(2)}_{k,M}&=\biggl(d^{(2)}_{k,M}+\frac{(d^{(1)}_{k,M})^2}{2}\biggr)
\frac{\Ai[C_{k}^{-1/3}(N+\frac{8}{k}-B_{k,M})]}
{\Ai[C_{k}^{-1/3}(N-B_{k,M})]},
\end{align}
and we have assumed that the worldsheet instantons are given by
\eqref{WSd},
\begin{align}
d^{(1)}_{k,M}
&=-\frac{n^0_{10}\beta^{-1}+n^0_{01}\beta}{4\sin^2\frac{2\pi}{k}},
\nonumber\\
d^{(2)}_{k,M}
&=\frac{n^0_{10}\beta^{-2}+n^0_{01}\beta^{2}}{8\sin^2\frac{4\pi}{k}}
+\frac{n^0_{20}\beta^{-2}+n^0_{11}+n^0_{02}\beta^{2}}{4\sin^2\frac{2\pi}{k}}.
\label{dkM}
\end{align}

From the asymptotic form of the Airy function
\begin{align}
\Ai[z]=\frac{e^{-\frac{2}{3}z^{3/2}}}{2\sqrt{\pi}z^{1/4}}
\biggl(1-\frac{5}{48}z^{-3/2}+{\mathcal O}(z^{-3})\biggr),
\end{align}
we find
\begin{align}
&\frac{\Ai[C_{k}^{-1/3}(N+\frac{4m}{k}-B_{k,M})]}
{\Ai[C_{k}^{-1/3}(N-B_{k,M})]}
=e^{-2\pi\sqrt{2\hat\lambda}m}
\biggl(1-\frac{2\sqrt{2}\pi m(m-\frac{1}{6})}{k^2\sqrt{\hat\lambda}}
-\frac{m}{k^2\hat\lambda}+{\mathcal O}(k^{-4})\biggr).
\end{align}
Hence, the free energy is given by
\begin{align}
F^{{\rm WS}(1)}_{k,M}
&=e^{-2\pi\sqrt{2\hat\lambda}}\biggl[g_s^{-2}\frac{1}{4}(n^0_{10}\beta^{-1}+n^0_{01}\beta)
+{\mathcal O}(g_s^0)\biggr],\nonumber\\
F^{{\rm WS}(2)}_{k,M}
&=e^{-4\pi\sqrt{2\hat\lambda}}\biggl[
g_s^{-2}\biggl(-\frac{1}{32}(n^0_{10}\beta^{-2}+n^0_{01}\beta^{2})
-\frac{1}{4}(n^0_{20}\beta^{-2}+n^0_{11}+n^0_{02}\beta^{2})
\nonumber\\
&\qquad\qquad+\frac{1}{16}(n^0_{10}\beta^{-1}+n^0_{01}\beta)^2x\biggr)
+{\mathcal O}(g_s^0)\biggr],
\end{align}
with $x=1/(\pi\sqrt{2\hat\lambda})$.

After plugging the Gopakumar-Vafa invariants \cite{GV,AMV},
\begin{align}
n^0_{10}=n^0_{01}=-2,\quad n^0_{20}=n^0_{02}=0,\quad n^0_{11}=-4,
\end{align}
this reproduces the genus-0 free energy
\begin{align}
F_{g=0}&=\frac{4\pi^3\sqrt{2}}{3}\hat\lambda^{3/2}
+\frac{2\pi^3 \ii}{3}\biggl(\frac{M}{k}\biggr)^3+{\rm const}\nonumber\\
&\hspace{-1cm}-\frac{1}{2}(\beta+\beta^{-1})e^{-2\pi\sqrt{2\hat\lambda}}
+\biggl(\frac{1}{16}(\beta^2+16+\beta^{-2})
+\frac{x}{4}(\beta+\beta^{-1})^2\biggr)e^{-4\pi\sqrt{2\hat\lambda}}
+{\mathcal O}(e^{-6\pi\sqrt{2\hat\lambda}}),
\end{align}
which was found in subsection 5.3 of \cite{DMP1}.

\subsection{Cancellation mechanism}\label{polecancel}
In the preceding subsections, we have presented a consistency check
with previous studies for the perturbative part and the worldsheet
instanton part of our conjecture \eqref{largemu}.
Note that these worldsheet instantons contain divergences at certain
coupling constants.
(See \eqref{dkM}.)
As in the case of the ABJM matrix model \cite{HMO2}, since there
should be no divergences in the matrix integration for $0\le M\le k$,
the divergences have to be cancelled by the membrane instantons and
the bound states.
Corresponding to the extra phases from $\beta^{\pm 1}$ in
$d_{k,M}^{(m)}$, we have found that the singularity of the worldsheet
instanton \eqref{WSd} is cancelled if we introduce the extra sign
factor $(-1)^{M\ell}$ in the membrane instantons.
Namely, we have checked that the singularity in
\begin{align}
d_{k,M}^{(m)}e^{-4m\mu_{\rm eff}/k}
+(-1)^{M\ell}\bigl(\widetilde b_{k}^{(\ell)}\mu_{\rm eff}
+\widetilde c_{k}^{(\ell)}\bigr)e^{-2\ell\mu_{\rm eff}},
\end{align}
at $k=2m/\ell$ is canceled for several values.
The extra sign factor $(-1)^{M\ell}$ can also be understood by the
shift of the Kahler parameters in the ABJ matrix model as pointed out
below \eqref{kahler}.

\section{Phase factor}\label{phasedepend}
After the consistency check of the perturbative sum, the worldsheet
instantons and the cancellation mechanism in the previous section,
let us start to compute the grand partition function $\Xi_{k,M}(z)$ in
\eqref{abjpf}.
Since the grand partition function $\Xi_{k,0}(z)$ of the ABJM matrix
model was studied carefully in our previous paper \cite{HMO2}, we
shall focus on the computation of the components of the matrix
\eqref{pfcomponent}.
After expanding in $z$, we find
\begin{align}
H_{m,n}(z)=\sum_{N=0}^\infty(-z)^NH_{m,n}^{(N)},
\end{align}
where each term $H_{m,n}^{(N)}$ is simply given by a $2N+1$ multiple
integration.

For $N=0$ we easily find ($\hbar=2\pi k$)
\begin{align}
H_{m,n}^{(0)}
=\int\frac{dy}{\hbar}
e^{\frac{2\pi}{\hbar}(m+\frac{1}{2})y}e^{-\frac{\ii}{2\hbar}y^2}e^{-\frac{2\pi}{\hbar}(n+\frac{1}{2})y}
=\frac{e^{-\frac{\pi \ii}{4}}}{\sqrt{k}}
e^{-\frac{2\pi \ii}{2k}(m-n)^2},
\end{align}
while for $N\ne 0$ we find
\begin{align}
H_{m,n}^{(N)}
&=\int\frac{dy_0}{\hbar}\frac{dx_1}{\hbar}\frac{dy_1}{\hbar}
\cdots\frac{dx_{N}}{\hbar}\frac{dy_{N}}{\hbar}
e^{\frac{2\pi}{\hbar}(m+\frac{1}{2})y_0}e^{-\frac{\ii}{2\hbar}y_0^2}
\frac{1}{2\cosh\frac{y_0-x_1}{2k}}e^{\frac{\ii}{2\hbar}x_1^2}
\frac{1}{2\cosh\frac{x_1-y_1}{2k}}
\nonumber\\
&\quad\times e^{-\frac{\ii}{2\hbar}y_1^2}
\cdots
\frac{1}{2\cosh\frac{y_{N-1}-x_N}{2k}}
e^{\frac{\ii}{2\hbar}x_N^2}
\frac{1}{2\cosh\frac{x_N-y_N}{2k}}
e^{-\frac{\ii}{2\hbar}y_{N}^2}e^{-\frac{2\pi}{\hbar}(n+\frac{1}{2})y_{N}}.
\label{HNdef}
\end{align}
Introducing the Fourier transformation,
\begin{align}
\frac{1}{2\cosh\frac{y_{i-1}-x_i}{2k}}
&=\int\frac{dp_i}{2\pi}\frac{e^{-\ii p_i(y_{i-1}-x_i)/\hbar}}{2\cosh\frac{p_i}{2}},
\nonumber\\
\frac{1}{2\cosh\frac{x_i-y_{i}}{2k}}
&=\int\frac{dq_i}{2\pi}\frac{e^{-\ii q_i(x_i-y_{i})/\hbar}}{2\cosh\frac{q_i}{2}},
\end{align}
and integrating over $y_1,x_1,\cdots,y_{N+1}$, we find
\begin{align}
H_{m,n}^{(N)}&=\frac{e^{-\frac{\pi \ii}{4}}}{\sqrt{k}}
e^{-\frac{2\pi \ii}{2k}(m+\frac{1}{2})^2}
e^{-\frac{2\pi \ii}{2k}(n+\frac{1}{2})^2}
\int\frac{dp_1dq_1}{2\pi\hbar}\cdots\frac{dp_Ndq_N}{2\pi\hbar}\nonumber\\
&\quad 
e^{-\frac{1}{\hbar}2\pi(m+\frac{1}{2})p_1}\frac{1}{2\cosh\frac{p_1}{2}}e^{\frac{\ii}{\hbar}p_1q_1}
\frac{1}{2\cosh\frac{q_1}{2}}e^{-\frac{\ii}{\hbar}q_1p_2}\cdots
e^{\frac{\ii}{\hbar}p_Nq_N}\frac{1}{2\cosh\frac{q_N}{2}}
e^{-\frac{1}{\hbar}2\pi(n+\frac{1}{2})q_N}.
\end{align}
Using further the formulas
\begin{align}
\int\frac{dp_1}{2\pi}
e^{-\frac{1}{\hbar}2\pi(m+\frac{1}{2})p_1}\frac{1}{2\cosh\frac{p_1}{2}}e^{\frac{\ii}{\hbar}p_1q_1}
&=\frac{1}{2\cosh\frac{q_1+2\pi \ii(m+\frac{1}{2})}{2k}},\nonumber\\
\int\frac{dp_i}{2\pi}
\frac{e^{\frac{\ii}{\hbar}p_i(q_i-q_{i-1})}}{2\cosh\frac{p_i}{2}}
&=\frac{1}{2\cosh\frac{q_{i-1}-q_i}{2k}},
\quad (i=2,3,\cdots,N-1)
\label{intamb}
\end{align}
to carry out the $p$-integrations, we finally arrive at the
expression
\begin{align}
H_{m,n}^{(N)}&=\frac{e^{-\frac{\pi \ii}{4}}}{\sqrt{k}}
e^{-\frac{2\pi \ii}{2k}(m+\frac{1}{2})^2}
e^{-\frac{2\pi \ii}{2k}(n+\frac{1}{2})^2}
\int\frac{dq_1}{\hbar}\frac{dq_2}{\hbar}\cdots\frac{dq_N}{\hbar}
\frac{1}{2\cosh\frac{q_1+2\pi \ii(m+\frac{1}{2})}{2k}}
\nonumber\\&\quad\times
\frac{1}{2\cosh\frac{q_1}{2}}
\frac{1}{2\cosh\frac{q_1-q_2}{2k}}
\frac{1}{2\cosh\frac{q_2}{2}}
\cdots
\frac{1}{2\cosh\frac{q_{N-1}-q_N}{2k}}
\frac{1}{2\cosh\frac{q_N}{2}}
e^{-\frac{1}{k}(n+\frac{1}{2})q_N}.
\end{align}

As in the case of the Wilson loops, we can express $H_{m,n}^{(N)}$
($N\ne 0$) as
\begin{align}
H_{m,n}^{(N)}=\frac{e^{-\frac{\pi \ii}{4}}}{\sqrt{k}}
e^{-\frac{2\pi \ii}{2k}(m+\frac{1}{2})^2}
e^{-\frac{2\pi \ii}{2k}(n+\frac{1}{2})^2}
\int\frac{dx}{\hbar}
\frac{1}{2\cosh\frac{x+2\pi \ii(m+\frac{1}{2})}{2k}}
\frac{1}{2\cosh\frac{x}{2}}
\phi_n^{(N-1)}(x),
\label{polescrossing}
\end{align}
where the functions $\phi_n^{(N)}(x)$ are defined by
\begin{align}
\phi_n^{(N)}(x)=\sqrt{2\cosh\frac{x}{2}}
\int\frac{dy}{\hbar}\rho^{N}(x,y)
\frac{e^{-\frac{1}{k}(n+\frac{1}{2})y}}{\sqrt{2\cosh\frac{y}{2}}},
\label{rhoN}
\end{align}
with
\begin{align}
\rho(x,y)=\frac{1}{\sqrt{2\cosh\frac{x}{2}}}
\frac{1}{2\cosh\frac{x-y}{2k}}\frac{1}{\sqrt{2\cosh\frac{y}{2}}}.
\end{align}
In \eqref{rhoN}, the multiplication among the density matrices
$\rho(x,y)$ is defined with a measure $1/\hbar$,
\begin{align}
\rho^N(x,y)=\int\frac{dz}{\hbar}\,\rho(x,z)\,\rho^{N-1}(z,y).
\end{align}
The functions $\phi_n^{(N)}(x)$ can be determined recursively by
\begin{align}
\phi_n^{(N)}(x)=\sqrt{2\cosh\frac{x}{2}}
\int\frac{dy}{\hbar}\rho(x,y)
\frac{\phi_n^{(N-1)}(y)}{\sqrt{2\cosh\frac{y}{2}}},
\end{align}
with the initial condition $\phi_n^{(0)}(x)=e^{-\frac{1}{k}(n+\frac{1}{2})x}$.

Note that, in \eqref{polescrossing}, the function
$1/\cosh\frac{x+2\pi\ii(m+\frac{1}{2})}{2k}$ has poles aligning on the
imaginary axis.
The pole with the smallest positive imaginary part is at
$x=\pi\ii(k-2(m+\frac{1}{2}))$ for $M$ in the range $0\le M<(k+1)/2$
since $m$ runs from $0$ to $M-1$.
Hence, for $M$ in this range, the relative position between the pole
and the real axis is the same as the ABJM case $M=0$ and we can trust
the formula \eqref{polescrossing} literally.
However, for $(k+1)/2\le M\le k$ the above pole comes across the real
axis and we need to deform the integration contour of
\eqref{polescrossing}, which is originally along the real axis, to the
negative imaginary direction.
This phenomenon and the contour prescription rule were already pointed
out in \cite{AHS}.
In their work, they proposed this prescription by requiring the
continuity at $M=(k+1)/2$ and the Seiberg duality.
They also checked that this prescription gives the correct values of
the partition function \eqref{abjpfdef} for small $N$ and $k$.
Our above analysis further pins down the origin of this deformation of
the integration contour.
The deformation comes from changing the integration variables from
\eqref{HNdef} to \eqref{polescrossing}.
For simplicity, hereafter, we shall often refer to the validity range
as $0\le M\le k/2$ instead of $0\le M<(k+1)/2$.

\subsection{Phase factor}
Unlike the case of the Wilson loops, the complex phase factor looks
very non-trivial and needs to be studied separately.
Using our Fermi gas formalism \eqref{abjpf}, we have found from
numerical studies that the phase factor is given by a rather simple
formula:
\begin{align}
\frac{1}{2\pi}\arg Z_k(N,N+M)=\frac{1}{8}M(M-2)
+\frac{1}{4}MN-\frac{1}{12k}(M^3-M).
\label{phase}
\end{align}
We have checked this formula numerically for $N=0,1,2,3$.
The results are depicted in figure \ref{fig:phase}.
As noted in the above paragraph, our numerical studies are valid not
only for $0\le M\le k/2$ but also slightly beyond $k/2$;
$0<M<(k+1)/2$.
In fact, we believe that our phase formula \eqref{phase} is valid for
the whole region of $0\le M\le k$ because we can show that this phase
reproduces a phase factor appearing in the Seiberg duality
\begin{align}
\frac{1}{2\pi}\arg\frac{Z_k(N,N+M)}{[Z_k(N,N+k-M)]^*}
=\frac{k^2}{24}+\frac{1}{12}+\frac{k(N-1)}{4}
\end{align}
as was conjectured in \cite{KWYphase} and further interpreted as a
contact term anomaly in \cite{CDFKS}. 

\begin{figure}[htb]
\begin{center}
\begin{tabular}{cc}
\resizebox{75mm}{!}{\includegraphics{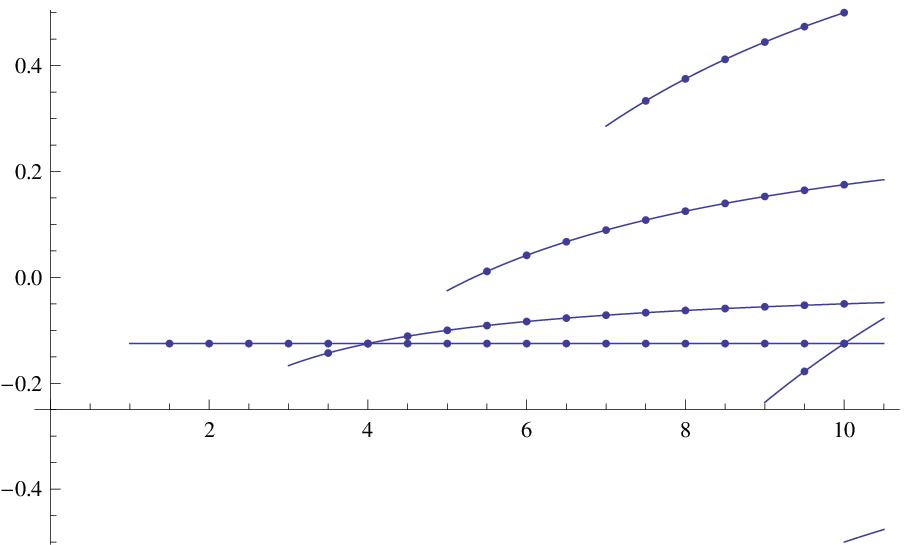}}
&
\resizebox{75mm}{!}{\includegraphics{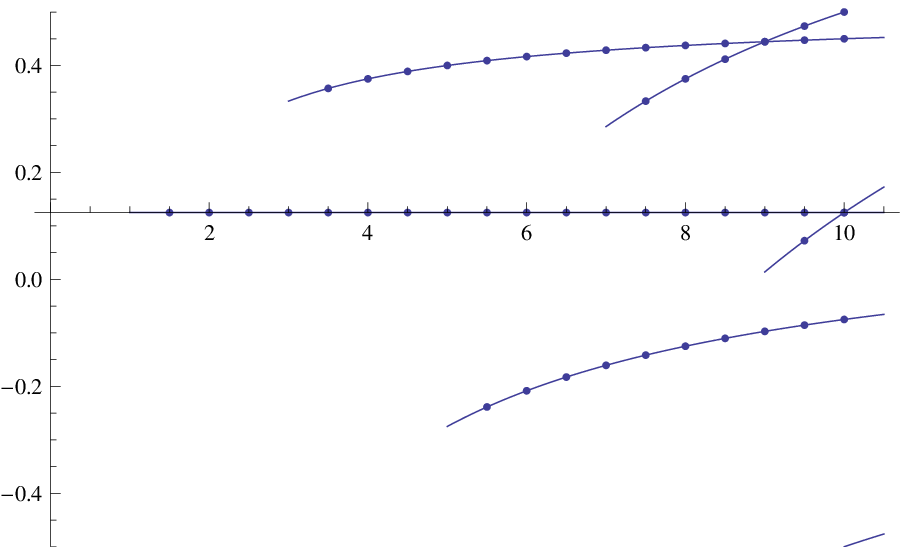}}
\\(a)\ $N=0$&(b)\ $N=1$\\
\resizebox{75mm}{!}{\includegraphics{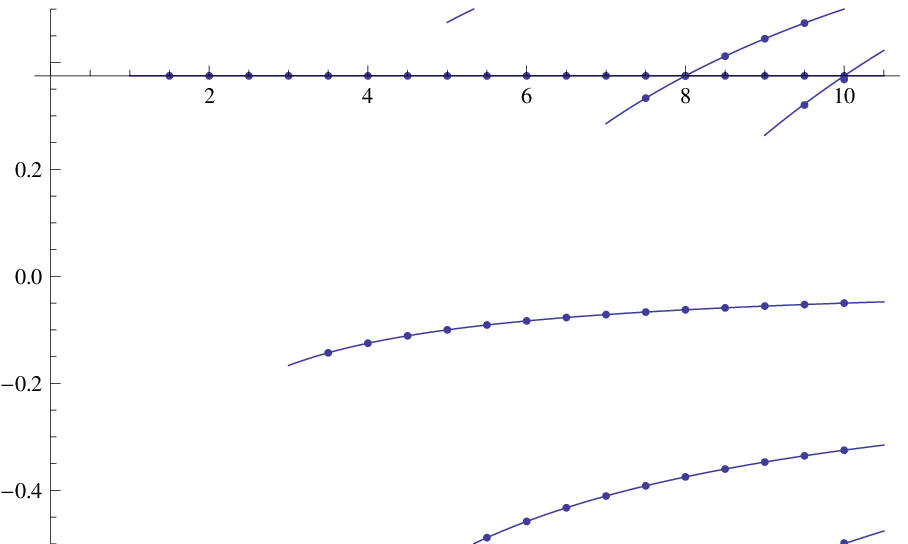}}
&
\resizebox{75mm}{!}{\includegraphics{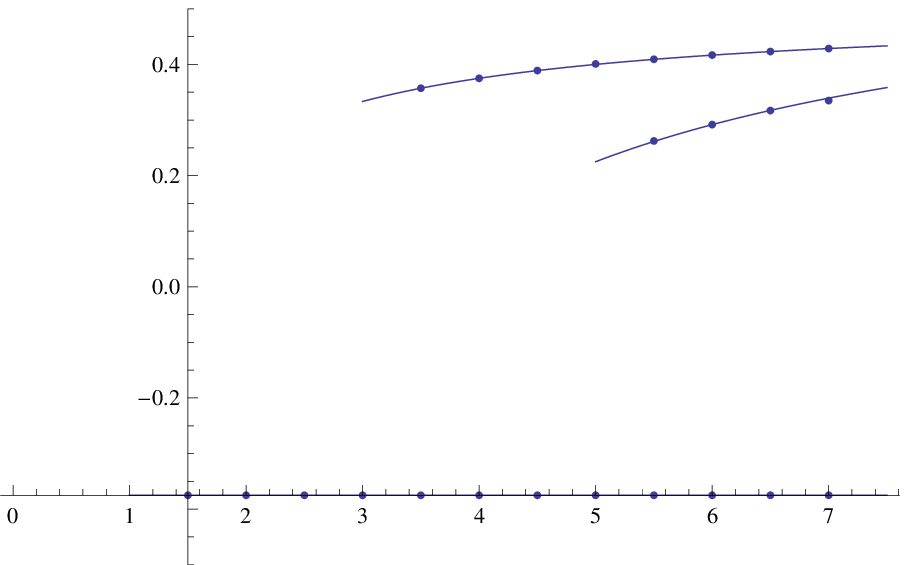}}
\\(c)\ $N=2$&(d)\ $N=3$
\end{tabular}
\end{center}
\caption{Numerical studies of the phase factor of the partition
function.
The horizontal axis denotes $k$ while the vertical axis shows the
phase normalized by $2\pi$.
Numerical data are depicted by points and our expectations
\eqref{phase} mod 1 are expressed by curves.
Each picture corresponds to different values of $N$ and each curve in
the picture starting from $k=2M-1$ corresponds different values of
$M$.
}
\label{fig:phase}
\end{figure}

\section{Grand Potential}\label{grandpotential}
After studying the phase factor of the partition function in the
previous section, let us turn to their absolute values and study the
grand potential defined by these absolute values \eqref{grandpot}.

\subsection{Grand potential at certain coupling constants}

As was found in \cite{HMO1,PY,HMO2} the computation of the ABJM
partition functions becomes particularly simple for $k=1,2,3,4,6$.
Also, as we have seen in section \ref{phasedepend}, the formula
\eqref{polescrossing} with integration along the real axis is
literally valid only for $0\le M\le k/2$.
Hence, we can compute various values of the partition function for
\begin{align}
(k,M)=(2,1),(3,1),(4,1),(6,1),(4,2),(6,2),(6,3).
\end{align}
The results of their absolute values are summarized in figure
\ref{values}.
\footnote{Some of the values were already found in \cite{Snote}.
Comparing our results with theirs is a very helpful check of our
formalism.
We are grateful to M.~Shigemori for sharing his unpublished notes with
us.}
As discussed in \cite{KWYphase}, the case of $k/2\le M\le k$ is
related to that of $0\le M\le k/2$ by the Seiberg duality.

\begin{figure}
\begin{align}
&|Z_2(0,1)|=\frac{1}{\sqrt{2}},\quad
|Z_2(1,2)|=\frac{1}{4\sqrt{2}\pi},\quad
|Z_2(2,3)|=\frac{\pi^2-8}{128\sqrt{2}\pi^2},\nonumber\\
&|Z_2(3,4)|=\frac{5\pi^2-48}{4608\sqrt{2}\pi^3},\quad
|Z_2(4,5)|=\frac{81\pi^4-848\pi^2+480}{294912\sqrt{2}\pi^4},\nonumber\\
&|Z_3(0,1)|=\frac{1}{\sqrt{3}},\quad
|Z_3(1,2)|=\frac{2-\sqrt{3}}{12},\quad
|Z_3(2,3)|=\frac{-(9\sqrt{3}-14)\pi+3\sqrt{3}}{432\pi},\nonumber\\
&|Z_3(3,4)|=\frac{14\pi-18-15\sqrt{3}}{1728\pi},\nonumber\\
&|Z_4(0,1)|=\frac{1}{2},\quad
|Z_4(1,2)|=\frac{\pi-2}{32\pi},\nonumber\\
&|Z_4(2,3)|=0.00003473909952494269119117566353230112859310233773233
\nonumber\\&\qquad
7261807934218890234955828380992634025931149937612,
\nonumber\\
&|Z_6(0,1)|=\frac{1}{\sqrt{6}},\quad
|Z_6(1,2)|=\frac{3\sqrt{3}-\pi}{108\sqrt{2}\pi},\nonumber\\
&|Z_6(2,3)|=3.76773027707758200049183186585155883429506373384028699
\nonumber\\&\qquad
96374213997516824024006754651401031813928511\times 10^{-6},
\nonumber\\
&|Z_6(3,4)|=5.26914099452731795482041046853051131744637477848566664
\nonumber\\&\qquad
22916096253100787064300949345207528685791\times 10^{-10},
\nonumber\\
&|Z_4(0,2)|=\frac{1}{2\sqrt{2}},\quad
|Z_4(1,3)|=\frac{4-\pi}{32\sqrt{2}\pi},\nonumber\\
&|Z_4(2,4)|=0.00001506227428345380302357520499270222421841701033492
\nonumber\\&\qquad
362553063511451195968480813607610027807404966983,\nonumber\\
&|Z_6(0,2)|=\frac{1}{6},\quad
|Z_6(1,3)|=\frac{7\pi-12\sqrt{3}}{432\pi},\nonumber\\
&|Z_6(2,4)|=4.77900663573206185466590506879892353173666149000261702
\nonumber\\&\qquad
495431896753514231026609667127826160173459\times 10^{-7},\nonumber\\
&|Z_6(0,3)|=\frac{1}{6\sqrt{2}},\quad
|Z_6(1,4)|=\frac{45\sqrt{2}-8\sqrt{6}\pi}{1296\pi},\nonumber\\
&|Z_6(2,5)|=2.34333487780752843368477720747976341731283580616750538
\nonumber\\&\qquad
345879256373591282194222350629426352014176\times 10^{-7}.\nonumber
\end{align}
\caption{Some exact or numerical values of partition functions.}
\label{values}
\end{figure}

Let us consider the grand potential defined with the absolute values
of the partition function \eqref{grandpot}.
Our strategy to determine the grand potential from the partition
function is exactly the same as that of \cite{HMO2} and we shall
explain only the key points here.
Since the grand potential with the sum truncated at finite $N$ always
contains some errors, it is known that fitting with the partition
function itself gives a result with better accuracy.
First we can compare the values found in figure \ref{values} with the
perturbative sum \eqref{pert}.
This already shows a good concordance.
For the $m$-th instanton effects, after subtracting the perturbative
sum and the major instanton effects, we fit the partition function
against the linear combinations of
\begin{align}
(-\partial_N)^nC_k^{-1/3}e^{A_k}
{\rm Ai}\Bigl[C_k^{-1/3}\Bigl(N+\frac{4m}{k}-B_{k,M}\Bigr)\Bigr].
\end{align}
Finally we reinterpret the result in terms of the grand potential.
Our results are summarized in figure \ref{grandpotkM}.

Compared with our study in \cite{HMO2,HMO3} we have much smaller
number of exact values of the partition function.
The lack of data causes quite significant numerical errors (about
1\%).
Nevertheless, since we have already known the rough structure of the
instanton expansion, we can find the exact instanton coefficient
without difficulty.

Note that the instanton coefficients of $(k,M)=(k,k/2)$ are similar to
those of $(k,M)=(k,0)$ for even $k$ and those of $(k,M)=(6,1), (6,2)$
are similar to those of $(k,M)=(3,1)$.
Due to this similarity, we have to confess that we only really fit the
values of the partition function for $(k,M)=(3,1)$ and $(k,M)=(4,1)$
up to seven instantons.
For other cases, after fitting for about three instantons, the
patterns become clear and we can bring the results from the known ones
and simply confirm the validity.

\begin{figure}
\begin{center}
\begin{align}
J^{\rm np}_{k=2,M=1}
&=\biggl[-\frac{4\mu^2+2\mu+1}{\pi^2}\biggr]e^{-2\mu}
+\biggl[-\frac{52\mu^2+\mu+9/4}{2\pi^2}+2\biggr]e^{-4\mu}\nonumber\\
&\quad+\biggl[-\frac{736\mu^2-304\mu/3+154/9}{3\pi^2}+32\biggr]e^{-6\mu}
\nonumber\\
&\quad+\biggl[-\frac{2701\mu^2-13949\mu/24+11291/192}{\pi^2}
+466\biggr]e^{-8\mu}\nonumber\\
&\quad+\biggl[-\frac{161824\mu^2-634244\mu/15+285253/75}{5\pi^2}
+6720\biggr]e^{-10\mu}\nonumber\\
&\quad+\biggl[-\frac{1227440\mu^2-5373044\mu/15+631257/20}{3\pi^2}
+\frac{292064}{3}\biggr]e^{-12\mu}+{\mathcal O}(e^{-14\mu}),
\nonumber\\
J^{\rm np}_{k=3,M=1}&=-\frac{2}{3}e^{-4\mu/3}-e^{-8\mu/3}
+\biggl[\frac{4\mu^2+\mu+1/4}{3\pi^2}-\frac{34}{9}\biggr]e^{-4\mu}
+\frac{25}{18}e^{-16\mu/3}+\frac{68}{15}e^{-20\mu/3}\nonumber\\
&\quad+\biggl[-\frac{52\mu^2+\mu/2+9/16}{6\pi^2}+\frac{296}{9}\biggr]e^{-8\mu}
-\frac{1894}{189}e^{-28\mu/3}+{\mathcal O}(e^{-32\mu/3}),\nonumber\\
J^{\rm np}_{k=4,M=1}
&=\biggl[\frac{4\mu^2+2\mu+1}{2\pi^2}-2\biggr]e^{-2\mu}
+\biggl[-\frac{52\mu^2+\mu+9/4}{4\pi^2}+18\biggr]e^{-4\mu}
\nonumber\\
&\quad+\biggl[\frac{736\mu^2-304\mu/3+154/9}{6\pi^2}-\frac{608}{3}\biggr]
e^{-6\mu}+{\mathcal O}(e^{-8\mu}),\nonumber\\
J^{\rm np}_{k=6,M=1}&=\frac{2}{3}e^{-2\mu/3}-e^{-4\mu/3}
+\biggl[-\frac{4\mu^2+2\mu+1}{3\pi^2}+\frac{34}{9}\biggr]e^{-2\mu}
+\frac{25}{18}e^{-8\mu/3}-\frac{68}{15}e^{-10\mu/3}\nonumber\\
&\quad
+\biggl[-\frac{52\mu^2+\mu+9/4}{6\pi^2}+\frac{296}{9}\biggr]e^{-4\mu}
+\frac{1894}{189}e^{-14\mu/3}+{\mathcal O}(e^{-16\mu/3}),\nonumber\\
J^{\rm np}_{k=4,M=2}
&=-e^{-\mu}+\biggl[-\frac{4\mu^2+2\mu+1}{2\pi^2}\biggr]e^{-2\mu}
-\frac{16}{3}e^{-3\mu}
+\biggl[-\frac{52\mu^2+\mu+9/4}{4\pi^2}+2\biggr]e^{-4\mu}
\nonumber\\&\quad
-\frac{256}{5}e^{-5\mu}
+\biggl[-\frac{736\mu^2-304\mu/3+154/9}{6\pi^2}+32\biggr]e^{-6\mu}
-\frac{4096}{7}e^{-7\mu}+{\mathcal O}(e^{-8\mu}),\nonumber\\
J^{\rm np}_{k=6,M=2}
&=-\frac{2}{3}e^{-\frac{2}{3}\mu}-e^{-\frac{4}{3}\mu}
+\biggl[\frac{4\mu^2+2\mu+1}{3\pi^2}-\frac{34}{9}\biggr]e^{-2\mu}
+\frac{25}{18}e^{-\frac{8}{3}\mu}+\frac{68}{15}e^{-\frac{10}{3}\mu}
\nonumber\\&\quad
+\biggl[-\frac{52\mu^2+\mu+9/4}{6\pi^2}+\frac{296}{9}\biggr]e^{-4\mu}
-\frac{1894}{189}e^{-\frac{14}{3}\mu}+{\mathcal O}(e^{-16\mu/3}),\nonumber\\
J^{\rm np}_{k=6,M=3}
&=-\frac{4}{3}e^{-\frac{2}{3}\mu}-2e^{-\frac{4}{3}\mu}
+\biggl[-\frac{4\mu^2+2\mu+1}{3\pi^2}-\frac{20}{9}\biggr]e^{-2\mu}
-\frac{88}{9}e^{-\frac{8}{3}\mu}-\frac{108}{5}e^{-\frac{10}{3}\mu}
\nonumber\\&\quad
+\biggl[-\frac{52\mu^2+\mu+9/4}{6\pi^2}-\frac{298}{9}\biggr]e^{-4\mu}
-\frac{25208}{189}e^{-\frac{14}{3}\mu}+{\mathcal O}(e^{-16\mu/3}).\nonumber
\end{align}
\end{center}
\caption{Grand potential obtained by fitting the exact or numerical
values of partition function.}
\label{grandpotkM}
\end{figure}

\subsection{Grand potential for general coupling constants}
Now let us compare the grand potential in figure \ref{grandpotkM} with
a natural generalization of our instanton expansion in the ABJM matrix
model.
We first observe a good match for the $m$-th pure worldsheet instanton
effects for $m<k/2$.
Secondly, we find that we have to modify signs by the factor
$(-1)^{M\ell}$ for the functions $a_k^{(\ell)}$,
$\widetilde b_k^{(\ell)}$, $\widetilde c_k^{(\ell)}$ characterizing
the membrane instantons.
This is important not only for ensuring the cancellation of the
divergences as we noted in subsection \ref{polecancel}, but also
for reproducing the correct coefficients of $\pi^{-2}$.
Thirdly, we confirm that the prescription of introducing the sign
factor $(-1)^{M\ell}$ reproduces correctly the bound states, where
there are no pure membrane instanton effects.

As for the constant term in the membrane instanton, there is an
ambiguity as long as it does not raise any singularities.
There are two candidates for it:
One is of course to take exactly the same constant term as in the ABJM
case when expressed in terms of the chemical potential $\mu$.
Another choice is to define $\widetilde c_k^{(\ell)}$ by respecting
the derivative relation.
Namely, in the ABJM matrix model it was observed that, when the grand
potential $J_k(\mu)$ is expressed in terms of the effective chemical
potential $\mu_{\rm eff}$, the constant term is the derivative of the
linear term \eqref{bc}.
These two choices give different answers because of the change in
$B_{k,M}$.
Comparing these two candidates with our numerical results in figure
\ref{grandpotkM}, we have found that neither of them gives the correct
answer.
Instead, the difference with the latter one is always $k/M$ times
bigger than the former one.
From this observation, we can write down a closed form for our
conjecture in \eqref{largemu}.
We have checked this conjecture up to seven worldsheet instantons and
four membrane instantons.

Although we restrict our analysis to the case $0\le M\le k/2$, we
believe our final conjecture \eqref{largemu} is valid for the whole
region of $0\le M\le k$ because of the consistency with the Seiberg
duality.
Though the expression \eqref{largemu} does not look symmetric in the
exchange between $M$ and $k-M$, if we pick up a pair of integers whose
sum is $k$, we find two identical instanton expansion series after
cancelling the divergences.\footnote{We are grateful to S.~Hirano,
K.~Okuyama, M.~Shigemori for valuable comments on it.}
We have checked this fact for all the pairs whose sums are
$k=1,2,3,4,6$.

\section{Discussions}\label{discussion}

In this paper we have proposed a Fermi gas formalism for the partition
function and the half-BPS Wilson loop expectation values in the ABJ
matrix models.
Our formalism identifies the fractional branes in the ABJ theory as a
certain type of Wilson loops in the ABJM theory.
Hence, our formalism shares the same density matrix as that of the
ABJM matrix model, which is suitable for the numerical studies. 
We have continued to study the exact or numerical values of the
partition function using this formalism.
Based on these values, we can determine the instanton expansion of the
grand potential at some coupling constants $k=2,3,4,6$ and conjecture
the expression \eqref{largemu} for general coupling constants.

Let us raise several points which need further clarifications.

The first one is the phase factor of our conjecture.
As we have seen in figure \ref{fig:phase}, we have checked this
conjecture for $N=0,1,2,3$ carefully.
However when $N\ge 3$ the numerical errors become significant and it
is difficult to continue the numerical studies with high accuracy for
large $k$.
It is desirable to study it more extensively.

The second one is the relation to the formalism of \cite{AHS}, which
looks very different from ours.
As pointed out very recently in \cite{H} it was possible to rewrite
the formalism of \cite{AHS} into a mirror expression where the
physical interpretation becomes clearer.
We would like to see the exact relation between theirs and ours.

Thirdly, we have found an extra term in \eqref{largemu} proportional
to the quantum mirror map $e_k^{(\ell)}$ \cite{HMMO}.
We have very few data to identify its appearance and it would be
great to check it also from the WKB expansion \cite{MP,CM}, though we
are not sure whether the restriction $0\le M\le k/2$ gives any
difficulty in the WKB analysis.
Furthermore, we cannot identify its origin in the refined topological
strings or the triple sine functions as proposed in \cite{HMMO}.
We hope to see its origin in these theories.
It may be a key to understand the gravitational interpretation
\cite{BGMS} of the membrane instantons.

The fourth one is about the Wilson loop in the ABJ theory.
After seeing that there are only new terms appearing in the membrane
instantons, we expect that the instanton expansion of the vacuum
expectation values of the Wilson loop should be expressed similarly
as that in the ABJM case \cite{HHMO}.
However, we have not done any numerical studies to support it.
Also, it is interesting to see how our study is related to other
recent works on the ABJ Wilson loops \cite{CGMS,BGLP,GMPS}.

Finally, one of the motivation to study the ABJ matrix model is its
relation to the higher spin models.
Since we have written down the grand potential explicitly, it is
possible to take the limit proposed in \cite{CMSY}.
We would like to see what lessons can be learned for the higher spin
models.

\section*{Acknowledgements}
We are grateful to Jaemo Park and Masaki Shigemori for very
interesting communications and for sharing their private notes with
us.
Sa.Mo.\ is also grateful to H.~Fuji, H.~Hata, Y.~Hatsuda, S.~Hirano,
M.~Honda, M.~Marino and K.~Okuyama for valuable discussions since the
collaborations with them.
The work of Sh.Ma.\ was supported by JSPS Grant-in-Aid for Young
Scientists (B) 25800062.

\appendix
\section{A useful determinantal formula}\label{detproof}
\begin{lemma}
Let $(\phi_i)_{1 \le i \le n+r}$ and $(\psi_j)_{1 \le j \le n}$ be functions
on a measurable space and let 
$(\zeta_{iq})_{\begin{subarray}{c} 1 \le i \le n+r \\ 1 \le q \le r \end{subarray}}$
be an array of constants. Then we have
\begin{align}
&\frac{1}{n!} \int \prod_{k=1}^n d x_k \cdot \det 
\left[ (\phi_i(x_k))
_{\begin{subarray}{c} 1 \le i \le n+r \\ 1 \le k \le n \end{subarray}}
\ (\zeta_{iq})
_{\begin{subarray}{c} 1 \le i \le n+r \\ 1 \le q \le r \end{subarray}} \right] \cdot
\det (\psi_j(x_k))_{1 \le j,k \le n } \nonumber\\
& = \det 
\left[ (m_{ij})
_{\begin{subarray}{c} 1 \le i \le n+r \\ 1 \le j \le n \end{subarray}}
\ (\zeta_{iq})
_{\begin{subarray}{c} 1 \le i \le n+r \\ 1 \le q \le r \end{subarray}} \right],
\label{detformula}
\end{align}
with
$m_{ij} = \int d x \phi_i(x) \psi_j(x)$.
\end{lemma}

\begin{proof}
Expand two determinants on the left hand side with respect to columns:
\begin{align}
& \frac{1}{n!} \int \prod_{k=1}^n d x_k 
\det \left[
(\phi_i(x_k))_{\begin{subarray}{c} 1 \le i \le n+r \\ 1 \le k \le n \end{subarray}} \ 
(\zeta_{iq})_{\begin{subarray}{c} 1 \le i \le n+r \\ 1 \le q \le r \end{subarray}} \right] \cdot
\det (\psi_j(x_k))_{1 \le j,k \le n }  \nonumber \\
&= \frac{1}{n!} \int \prod_{k=1}^n d x_k 
\sum_{\sigma \in S_{n+r}} \sgn(\sigma)\prod_{k=1}^n 
\phi_{\sigma(k)}(x_k) \cdot \prod_{q=1}^r \zeta_{\sigma(n+q),q} 
\sum_{\tau \in S_{n}} \sgn (\tau) \prod_{k=1}^n \psi_{\tau(k)}(x_k) \nonumber\\
&= \frac{1}{n!} \sum_{\sigma \in S_{n+r}} \sgn(\sigma) 
\prod_{q=1}^r \zeta_{\sigma(n+q),q} \cdot \sum_{\tau \in S_n} \sgn(\tau)
\prod_{k=1}^n \int dx \phi_{\sigma(k)}(x) \psi_{\tau(k)}(x)\nonumber\\
&= \frac{1}{n!} \sum_{\tau \in S_n} \sgn (\tau) \sum_{\sigma \in S_{n+r}} \sgn(\sigma) 
\prod_{q=1}^r \zeta_{\sigma(n+q),q} \cdot \prod_{k=1}^n m_{\sigma(k), \tau(k)} \nonumber\\
&= \frac{1}{n!} \sum_{\tau \in S_n} \sgn(\tau) \det 
\left[ 
(m_{i,\tau(j)})_{\begin{subarray}{c} 1 \le i \le n+r \\ 1 \le j \le n \end{subarray}} \ 
(\zeta_{iq})_{\begin{subarray}{c} 1 \le i \le n+r \\ 1 \le q \le r \end{subarray}} \right].
\end{align}
It follows from the alternating property for determinants that this equals to
\begin{equation}
\frac{1}{n!} \sum_{\tau \in S_n} 
\det 
\left[ 
(m_{i,j})_{\begin{subarray}{c} 1 \le i \le n+r \\ 1 \le j \le n \end{subarray}} \ 
(\zeta_{iq})_{\begin{subarray}{c} 1 \le i \le n+r \\ 1 \le q \le r \end{subarray}} \right]
=\det 
\left[ 
(m_{i,j})_{\begin{subarray}{c} 1 \le i \le n+r \\ 1 \le j \le n \end{subarray}} \ 
(\zeta_{iq})_{\begin{subarray}{c} 1 \le i \le n+r \\ 1 \le q \le r \end{subarray}} \right].
\end{equation}
\end{proof}

\section{Expansion of Fredholm determinant}\label{fredholm}

Although we have used an infinite-dimensional version, we shall
give a finite-dimensional version of the identity below.
For a positive integer $n$, we let $[n]=\{1,2,\dots,n\}$.

\begin{lemma}
Let $N, L$ be non-negative integers.
Let $A=(a_{ij})$, $B=(b_{iq})$, $C=(c_{pj})$,
and $D=(d_{pq})$ be matrices of finite sizes
$N \times N$, $N \times L$, $L \times N$, and $L \times L$,
respectively.
Let $\mathbf{1}'_{NL}$ be the 
$(N+L) \times (N+L)$ diagonal matrix
whose the first $N$ diagonal entries are $1$ and 
other entries are $0$. 
Then the following identity holds.
\begin{equation}
\det \left( \mathbf{1}'_{NL}+ \begin{pmatrix} A & B \\ C & D \end{pmatrix}
\right)
=\sum_{n=0}^N \frac{1}{n!} \sum_{k_1,\dots,k_n=1}^N  \det 
\begin{pmatrix} (a_{k_i, k_j})_{1 \le i,j \le n}  & 
(b_{k_i, q})_{1 \le i \le n, 1 \le q \le L} \\
(c_{p, k_j})_{1 \le p \le L, 1 \le j \le n} 
& D \end{pmatrix}.
\label{eq:fredholm}
\end{equation}
\end{lemma}

\begin{proof}
Put $\mathbf{A}= (\mathbf{a}_{ij})_{1 \le i,j \le N+L}=
 \begin{pmatrix} A & B \\ C & D \end{pmatrix}$.
Expanding the determinant with respect to rows, we have
\begin{equation}
\det ( \mathbf{1}'_{NL}+ \mathbf{A})
=
\sum_{\sigma \in S_{N+L}} \sgn (\sigma) 
\prod_{i=1}^N (\delta_{i,\sigma(i)}+\mathbf{a}_{i, \sigma(i)}) \times \prod_{p=1}^L
\mathbf{a}_{N+p, \sigma(N+p)}.
\end{equation}
Divide the product for $i$:
for each $\sigma \in S_{N+L}$, 
\begin{equation}
\prod_{i=1}^N (\delta_{i,\sigma(i)}+\mathbf{a}_{i, \sigma(i)})
= \sum_{I \subset [N]} \prod_{i \in I} \mathbf{a}_{i, \sigma(i)}
\times \prod_{i \in [N] \setminus I} \delta_{i,\sigma(i)}.
\end{equation}
Here the product $\prod_{i \in [N] \setminus I} \delta_{i,\sigma(i)}$
vanishes unless $\sigma(i)=i$ for all $i \in  [N] \setminus I$,
i.e., unless the support $\mathrm{supp}(\sigma)$
of $\sigma$ is a subset of $I \cup \{N+1,\dots,N+L \}$.
In that case, the permutation 
$\sigma$ can be seen as a permutation on 
$I \cup \{N+1,\dots, N+L\}$.
Denoting by $S_{I \cup \{N+1,\dots, N+L\}}$ the
permutation group consisting of such permutations,
\begin{align}
\det ( \mathbf{1}'_{NL}+ \mathbf{A})
=& \sum_{I \subset [N]} \sum_{\sigma \in S_{I \cup
\{N+1,\dots,N+L\}}}
\sgn(\sigma) \prod_{i \in I} \mathbf{a}_{i, \sigma(i)}
\times \prod_{p=1}^L
\mathbf{a}_{N+p, \sigma(N+p)} \nonumber\\
=& \sum_{I \subset [N]} \det 
\begin{pmatrix} (\mathbf{a}_{i,j})_{i,j \in I} &
(\mathbf{a}_{i,N+q})_{i \in I, q \in [L]} \\
(\mathbf{a}_{N+p,j})_{p \in [L], j \in I} & 
(\mathbf{a}_{N+p,N+q})_{p,q \in [L]} \end{pmatrix}.
\end{align}
It is immediate to see that
this identity presents the desired identity.
\end{proof}

\end{document}